\documentclass[11pt]{llncs}
\usepackage{color}
\usepackage{graphicx}
\usepackage{fullpage}
\usepackage{amsfonts}
\usepackage{latexsym,amssymb,amsmath}

\usepackage[ruled]{algorithm}
\floatname{algorithm}{Procedure}

\usepackage[noend]{algpseudocode}
\algdef{SE}[DOWHILE]{Do}{doWhile}{\algorithmicdo}[1]{\algorithmicwhile\ #1}
\usepackage{comment}
\usepackage{eqparbox}

\newcommand\LONGCOMMENT[1]{%
  \textbf{\#}\ \begin{minipage}[t]{0.8\linewidth}#1\strut\end{minipage}%
}

\usepackage{epstopdf}
\usepackage{soul}
\usepackage{enumitem}

\newcommand{\D}{\mbox{done}}
\newcommand{\HandleVertex}{\textsc{HandleVertex}}
\newcommand{\HandleEdge}{\textsc{HandleEdge}}

\begin{document}

\title{Graph Exploration by Energy-Sharing Mobile Agents\thanks{This is the full version of the paper appearing in the proceedings of SIROCCO 2021.}}
\author{
J. Czyzowicz\inst{1}
\and
S. Dobrev\inst{2}
\and
R. Killick\inst{3}
\and
E. Kranakis\inst{3}
\and
D. Krizanc\inst{4}
\and 
L. Narayanan\inst{5}
\and 
J. Opatrny\inst{5}
\and
D. Pankratov\inst{5}
\and 
S. Shende\inst{6}
}
\institute{D\'epartement d'Informatique, Universit\'e du Qu\'ebec en Outaouais,
Canada
\and
Slovak Academy of Sciences, Bratislava, Slovakia
\and
School of Computer Science, Carleton University, Ottawa, Canada
\and
Dept.of Mathematics \& Computer Science, Wesleyan University, Middletown 
CT, USA
\and
Department of Computer Science  and Software Engineering, Concordia University, Canada
\and
Department of Computer Science, Rutgers University, USA
}
\maketitle

{\bf Abstract} 
We consider the problem of collective exploration of a known $n$-node edge-weighted graph by $k$ mobile agents that have limited energy but are capable of energy transfers. The agents are initially placed at an arbitrary  subset of nodes in the graph, and each agent has an initial, possibly different, amount of energy. The goal of the exploration problem is for every edge in the graph to be traversed by at least one agent. The amount of energy used by an agent to travel distance $x$ is proportional to $x$. In our model, the agents can {\em share} energy when co-located:  when two agents meet, one can transfer part of its energy to the other.

For an $n$-node path, we give an
$O(n+k)$ time algorithm that either finds an exploration strategy, or reports that one does not exist. 
For an $n$-node tree with $\ell $ leaves, we give an 
$O(n+ \ell k^2)$ algorithm that finds an exploration strategy if one exists.  Finally, for the general graph case, we show that the problem of deciding if  exploration is possible by energy-sharing agents is NP-hard,  even for 3-regular graphs. In addition, we show that it is always possible to find an exploration strategy if the total energy of the agents is at least twice the total weight of the edges; moreover, this is asymptotically optimal.


\vspace{0.5cm}



\section{Introduction}

The emergence of swarm robotics has inspired a number of investigations into the capabilities of a collection of autonomous mobile  robots (or agents), each with limited capabilities. Such agents cooperate and work collaboratively  to achieve complex tasks such as pattern formation, object clustering and assembly, search, and exploration. Collaboration on such tasks is achieved by, for example, decomposing the task at hand into smaller tasks which can be performed  by individual agents. The benefits of the collaborative paradigm are manifold: smaller task completion time, fault tolerance, and the lower build cost and energy-efficiency of a collection of smaller agents as compared to larger more complex agents. Somewhat surprisingly, for example,  a recent paper \cite{jurek2019groupsearch} shows that two agents can search for a target at an unknown location on the line with lower total energy costs than a single agent.

In this paper, we study the problem of collective exploration of a known edge-weighted graph by $n$ mobile agents initially placed at arbitrary nodes of the graph. Many variants of the graph exploration problem have been studied previously; see Section~\ref{sec:related-work} for a description of some of the related work. For our work, the goal of exploration is that  {\em every edge} of the graph must be traversed by at least one agent. 
The weight of an edge is called its {\em length}. Every agent is equipped with a  {\em battery/energy container} that has an initial amount of energy; the initial energies of different agents can be different. We assume that moving length $x$  depletes the battery of an agent by exactly $x$.

Clearly then, for exploration to be possible, the sum of the initial energies of all agents has to be at least $\cal{E}$, the sum of all edge weights. However total energy $\cal{E}$ may not be sufficient; the {\em initial placement} of the agents plays a role in deciding if exploration is possible with the given energies.  To see this, consider exploration by 2 agents of a path with 4 nodes, where each of the 3 edges has length 1. If the agents are initially placed at the two endpoints of the path, then total energy $3$ suffices to explore the path. However if the two agents are initially placed at the middle two nodes  of the path, it is not  difficult to see that total energy 4 is necessary to complete the exploration.

In addition to initial placement of agents, and the total amount of energy, the initial {\em energy distribution} also affects the existence of an exploration strategy. To see this,  suppose the 2 agents are placed at the middle nodes of the 4-node path. Consider first an energy distribution in which  both agents have initial energy 2. Then one exploration strategy would be for  both agents to explore half of the center edge, and turn around to travel to the endpoint.
Next consider an energy distribution in which agent 1 has energy $3+\epsilon$ for some $0 < \epsilon \leq 1$ and the agent 2 has energy $1-\epsilon$. It is easy to see that exploration is impossible, even though the total energy of both agents is the same as in the first distribution. 

Recently, several researchers have proposed a new mechanism to aid collaboration:  {\em the capability to share energy}.  In other words, when two agents meet, one can transfer a portion of its energy to the other. It is interesting to investigate what tasks might be made possible with this new capability, given the same initial amounts of energies. In \cite{bampas2017collaborative,czyzowicz2016communication,czyzowicz2018broadcast,czyzowicz2019energy,moussi2018data}, researchers have studied the problems of data delivery, broadcast, and convergecast by energy-sharing mobile agents.

In the example described above, where agent 1 has energy $3+\epsilon$ for some $0 < \epsilon \leq 1$ and the agent 2 has energy $1-\epsilon$,  if energy transfer is allowed, agent 1 (with the higher energy) can first go to the endpoint closer to its initial position, then turn around, reach agent 2, and transfer its remaining energy $\epsilon$ to agent 2. This enables agent 2 to reach the other endpoint, thereby completing the exploration. 

This simple example shows that energy-sharing capabilities make graph exploration possible  in situations where it would have been impossible otherwise. Note that an algorithm for exploration with energy sharing requires not only an assignment of trajectories to agents that collectively explore the entire graph, but also an achievable schedule of energy transfers. In  this paper, we are interested in exploration strategies for edge-weighted graphs by energy-sharing mobile agents. We give a precise definition of our model and the collaborative exploration problem below.

\subsection{Model}

We are given a {\em weighted graph} $G= (V, E)$ where $V$ is a set of $n$
 vertices (or nodes), $E$ a set of $m$ edges, and each edge $a_i \in E$ is assigned a real number  $w_i$, denoting its {\em length}.
We have  $k$ {\em mobile agents} (or robots) $r_1, r_2, \ldots, r_{k}$ placed at some of the vertices of the graph. We allow more than one agent to be located in the same place.  Each mobile agent (or agent for short) $r_i$ can move with speed 1, and  initially possesses a specific amount of {\em energy} equal to $e_i$ for its moves.
An agent can move in any direction along the edges of the graph $G$, it can stop
if needed, and it can reverse its direction of moving either at a vertex, or after traversing a part of an edge. The energy consumed by a moving agent is linearly proportional to the distance $x$ traveled; to simplify notation it is assumed to be equal to $x$.
An agent can move only if its energy is greater than zero.

An important feature of our model is the possibility of {\em energy sharing} between agents: when two agents, say $r_i$ and $r_j$, $i\neq j$, meet at some time at some location in the graph, agent $r_i$ can transfer  a portion of its energy to $r_j$.
More specifically, if $e_i'$ and $e_j'$ are the energy levels of  $r_i$ and $r_j$
at the time they meet
then  $r_i$ can transfer to $r_j$ energy $0< e\leq e_i'$ and thus their energies will become  $e_i' -e$ and  $e_j'+e$, respectively.

In our model, each agent is assigned a {\em trajectory} to follow.
We define a {trajectory} of an agent to
be a sequence of edges or parts of edges that starts at the agent's initial position and forms a continuous walk in the graph. In addition, a trajectory specifies
a {\em schedule} of energy transfers, i.e., all
points on this walk (could be  points different from  vertices) where the agent is to receive/transfer energy from/to other agents, and for each such point the amounts of energy involved.
We call a set of trajectories {\em valid} if the schedules of energy transfers among trajectories match, and energy levels are sufficient for the movement of agents. More specifically, for every transfer point on a trajectory of agent $r_i$ where energy  is to be received/transferred,
 there is exactly one  agent $r_j$, $j \neq i$, whose trajectory contains the same transfer point transferring/receiving
 that amount of energy to/from $r_i$, and the transfers can scheduled on a time
line. Furthermore, the energy of an agent, initially and after any energy transfer,  must be  always sufficient to continue to move along its assigned trajectory. We are interested in solving the following general problem of collaborative exploration:
 
\vspace*{2mm}
\noindent
{\bf Graph Exploration Problem}: Given a {\em weighted graph} $G= (V, E)$ and $k$ mobile agents $r_1,$ $r_2, \ldots ,r_{k}$ together with their respective initial energies $ e_1,e_2,\ldots, e_{k}$ and  positions $s_1,s_2,\ldots, s_k$ in the graph, find a valid set of trajectories that {\em explore} (or cover) all edges of the graph.

\subsection{Related work}
\label{sec:related-work}

The problems of exploration and searching have been investigated for over fifty years. The studied environments were usually graphs (e.g.\cite{AH00,deng1999exploring,koutsoupias1996searching,panaite1999exploring}) and geometric two-dimensional terrains (e.g.\cite{albers2002exploring,baezayates1993searching,deng1991learn}). The goal of such research was most often the minimization of the time of the search/exploration that was proportional to the distance travelled by the searcher. The task of searching consists of finding the target placed at an unknown position of the environment. The environment itself was sometimes known in advance (cf. \cite{baezayates1993searching,beck1964linear,czyzowicz2019energy,koutsoupias1996searching}) but most research assumed only its partial knowledge, e.g. the type of graph, the upper bound on its size or its node degree, etc. Remarkably, there exist hundreds of papers for search in an environment as simple as a line (cf. \cite{alpern2003theory}). The task of exploration consisted of constructing a traversal of the entire environment, e.g. in order to construct its map (see \cite{kleinberg1994line,panaite1999exploring}). It is worth noting that performing a complete graph traversal does not result in acquiring the knowledge of the map (see \cite{chalopin2010constructing}).

Most of the early research on search and exploration has been done for the case of a single searcher. When a team of collaborating searchers (also called agents or robots) is available, the main challenge is usually to partition the task among the team members and synchronize their efforts using available means of communication, cf. \cite{baeza1995parallel,burgard2005coordinated,dereniowski2015fast,fraigniaud2006collective}. Unfortunately, for the centralized setting, already in the case of two robots in the tree environment known in advance, minimizing its exploration time is NP-hard, e.g., see \cite{fraigniaud2006collective}.

The case of robots that can share energy has been recently studied for the tasks of data communication, \cite{bampas2017collaborative,czyzowicz2016communication,czyzowicz2018broadcast,czyzowicz2019energy,moussi2018data}. In this research the robots are distributed in different places of the network, each robot initially possessing some amount of energy, possibly distinct for different robots. The energy is used proportionally to the distance travelled by the robot. The simplest communication task is {\em data delivery} (see \cite{bampas2017collaborative,bartschi2017efficient,chalopin2014data,czyzowicz2016communication}), where the data packet originally placed in some initial position in the environment has to be carried by the collaborating robots into the target place. Remarkably, when the robots cannot share energy, data delivery is an NP-hard problem even for the line network, (see \cite{chalopin2014data}). When the robots are allowed to exchange a portion of energy while they meet in the tree of $n$ nodes, \cite{czyzowicz2016communication} gives the $O(n)$-time solution for the data delivery. For energy sharing robots, the authors of \cite{czyzowicz2018broadcast} study the {\em broadcast problem}, where a single packet of data has to be carried to all nodes of the tree network, while \cite{czyzowicz2016communication} investigates also the {\em convergecast problem}, where the data from all tree nodes need to be accumulated in the memory of the same robot. In both cases efficient communication algorithms are proposed. A byproduct of \cite{czyzowicz2019energy} is an optimal exploration algorithm in the special case when all robots are initially positioned at the same node of the tree. When the energy sharing robots have small limited memory, able to carry only one or two data packets at a time, the simplest case of the data delivery problem is shown to be NP-hard in \cite{bampas2017collaborative}. Further, in~\cite{dynia2007robots} bounds are proved in an energy model where robots can communicate when they are in the same node and the goal of robot team is to jointly explore an unknown tree.

\subsection{Results of the paper}

We start in Section~\ref{sec:Exploring a Path} with exploration of a path.    
Given an initial placement and energy distribution for $k$ energy-sharing agents on an $n$-node path, we give an $O(n+k)$ algorithm to generate a set of valid 
trajectories  whenever the exploration of the path is possible. We also show that a path can always be explored if the total energy of energy-sharing agents  
is $\frac 3 2$ times the total weight of edges in the path. In contrast, we show that there are energy configurations for which any total amount of energy is insufficient for 
path exploration {\em without} energy sharing. 

In Section~\ref{sec:Exploring a Tree} we study exploration of trees. We first observe 
that without energy sharing the exploration of trees is NP-complete.
Then, for an $n$-node tree, we give an 
$O(n+ \ell k^2)$ algorithm that finds an exploration strategy if one exists, where $\ell $ is the number of leaves in the tree. 


In Section~\ref{sec:General Graphs}, we consider exploration of general graphs. We show that the problem is NP-hard even for 3-regular graphs. In addition, we show that it is always possible to find an exploration strategy if the total energy of the agents is at least twice the total weight of the edges; moreover, this is asymptotically optimal, even for trees. 


Therefore our results show that allowing energy to be shared between agents makes exploration possible in many situations when it would not be possible without sharing energy. Furthermore, the total energy needed for exploration is at most twice (at most 3/2) the total weight of the edges in the graph (path respectively), while there is no upper bound on the total energy needed for exploration if agents cannot share energy, even when the graph to be explored is a path. 





%

\section{Exploring a Path}
\label{sec:Exploring a Path}

In this section we consider the case when the graph is a simple path on $n$ nodes; without loss of generality, we assume that the path is embedded in the
horizontal line segment $[0,1]$, and we will refer to the movements of agents in their trajectories as being  left/right movements on the segment.  
 Clearly, in case the graph is given in the usual graph representation, this embedding can be obtained in $O(n+k)$ time.
The path exploration problem can therefore be restated  as follows:
\noindent
\begin{problem}[Segment Exploration]
\label{problem:p1}
Given mobile agents $r_1,r_2,\ldots,r_k$ with energies  $e_1,e_2,\ldots,e_k$, located initially in positions
$0\leq s_1 \leq s_2 \leq \cdots \leq s_k\leq 1$ of a line segment  $[0, 1]$, respectively,  find a set of valid trajectories of these agents that explore the segment, if possible.
\end{problem}

A trajectory $t_i$ of  agent $r_i$ explores  a closed sub-segment $a_i$ of $[0,1]$ containing  $s_i$. Let $b_i^{\ell}$, $b_i^r$ be the left, right end point of this sub-segment. We want to find a valid set of trajectories, i.e.,  a set that explores the line segment $[0, 1]$, and  there exists a schedule for energy transfers such that every agent has enough energy to follow its trajectory.

We first observe that in the case of exploring a line segment with the possibility of energy sharing some assumptions on the shape of valid trajectories can be made without loss of generality.

\vspace*{-3mm}
\begin{enumerate}
\setlength{\itemsep}{0pt}

\item  The segments $a_1,a_2,\ldots,a_k$ explore (or cover) $[0,1]$ and they don't overlap, i.e., $b_1^{\ell} = 0$, $b_n^r = 1$, and $b_{i}^r=b_{i+1}^{\ell}$ for
$1\leq i \leq k-1$.
\item  Trajectory $t_i$ starts at $s_i$, goes  straight to one of the endpoints of $a_i$. When both endpoints are different from $s_i$, it
 turns around and goes straight the other endpoint of $a_i$.
Thus,  in this case the trajectory {\it covers doubly} a sub-segment between $s_i$ and the endpoint where it turns around, and the trajectory has a {\em doubly covered part} and a {\em singly covered part}.
 \item A transfer of energy between two agents $r_i$ and $r_{i+1}$ may occur only at  their {\em meeting point} $b_i^r$. Thus at  $b_i^r$ exactly one of the following occurs:

\begin{enumerate}
\item There is no energy transfer.
\item  There is  energy transfer from $r_i$ to $r_{i+1}$. In that case $t_{i+1}$ does not end at that point,  it ends at $b_{i+1}^r$, and either $b_{i+1}^{\ell}=s_{i+1}$ or $b_{i+1}^{\ell}$
is a point where the trajectory $t_{i+1}$ turns around to the right.
\item There is energy transfer from $r_{i+1}$ to $r_i$. In that case $t_{i}$ does not end at that point, it ends at $b_{i}^{\ell}$ and either $b_i^r=s_i $ or $b_i^r$  is a point where the trajectory $t_{i}$ turns around to the left.
\end{enumerate}
\end{enumerate}

The next lemma,  stated without proof, 
specifies two additional restrictions that can be imposed on the nature of valid trajectories that will be applied by our algorithm.

\begin{lemma}
\label{lemma:l1}
Assume that the segment $[0,1]$ can be explored by a set of valid trajectories $T=\{t_1,t_2,\ldots, t_k\}$ of the agents. Then there is a {\em canonical} set of valid trajectories  $T'=\{t'_1,t'_2,\ldots, t'_k\}$  that explore the  segment
such that
\begin{enumerate}[label=(\roman*)]
\item If agent $r_i$  receives energy from a right (left) neighbour then it receives it at its initial position $s_i$, and its trajectory may only go in straight line segment from $s_i$ to the left (right).
\item For  each trajectory, its singly covered part is at least as long as its doubly covered part.
\end{enumerate}
\end{lemma}

\begin{algorithm}
\caption{\textsc{Path}$(i,\ell_i,tr_i)$}
\label{alg-path}
{\normalsize
\begin{algorithmic}[1]
\If {$tr_i\leq 0$} 
	\If{ $e_i< s_i-\ell_i +tr_i$} \Comment{\textbf{Case 1.1} (deficit increases)}		
\State \LONGCOMMENT{$r_i$ waits to receive energy $|tr_{i+1}|$ 
		from $r_{i+1}$. Trajectory $t_i$ is from $s_i$ to $\ell_i$.  Transfer $|tr_i|$ of energy to $r_{i-1}$}
		\State $\ell_{i+1} \gets  s_i$  
\State $tr_{i+1} \gets e_i +tr_i -(s_i-\ell_i)$ 
	\ElsIf {$e_i \geq s_i-\ell_i+tr_i$} \Comment{\textbf{Case 1.2} (deficit eliminated)}
	\State \LONGCOMMENT{Using the values of $\ell_i$, $s_i$ and   $e_i-|tr_i|$, select a canonical trajectory $t_i$ originating in  $s_i$, located 
between $\ell_i$ and $s_{i+1}$ as in {\bf Cases 1.2.1} to {\bf 1.2.3}. Transfer $|tr_i|$ of energy to $r_{i-1}$}
\State $tr_{i+1} \gets e_i +tr_i -length(t_i)$
\State $\ell_{i+1} \gets right\_endpoint(t_i)$
	\EndIf
\Else \Comment{\textbf{Case 2}: (Energy surplus)}
	\State \LONGCOMMENT{A surplus of energy at $\ell_i$ implies that  $\ell_i=s_i$. 	The trajectory $t_i$ of $r_i$ is ``from $s_i$ to the right, but at most to $s_{i-1}$''.}
	\State $\ell_{i+1} \gets s_i+\min\{s_{i+1}-s_i,s_i+e_i+tr_i\}$
\State  $tr_{i+1}\gets tr_i +e_i-(\min\{s_{i+1}-s_i,s_i+e_i+tr_i\}-s_i))$
\EndIf
\If {$(i <k)$}
	return {\textsc{Path}$(i+1,\ell_{i+1},tr_{i+1})$}
\ElsIf {$(\ell_{k+1} <1)$ or $(tr_{k+1} < 0)$}
\Comment{Insufficient energy}
	\State return with \textbf{solvable} $\gets$ false
\Else 
	\State return with \textbf{solvable} $\gets$ true
\EndIf
\end{algorithmic}
}
\end{algorithm}
We now describe a recursive, linear time algorithm for Problem \ref{problem:p1} to find canonical trajectories as described in Lemma \ref{lemma:l1} above. The trajectories are assigned to agents from left to right, determining whether more energy needs to be transferred to complete the coverage on the left, or some surplus  
energy is to be transfused to agents on the right.   If $i$ is the index of the leftmost agent not used in the exploration of initial sub-segment $[0, \ell_i]$ and if $tr_i$ is the energy deficit (negative) or surplus (positive) left after exploring this sub-segment using only agents 1 through $(i-1)$, then the procedure call \textsc{Path}$(i, \ell_i, tr_i)$ decides whether a solution to the exploration problem for the remaining sub-segment $[\ell_i, 1]$ is possible via canonical trajectories. It does so by greedily deploying agent $i$ to use the least amount of energy to cover at least the segment $[\ell_i, s_i]$: 
if this does not lead to energy deficit, then the trajectory of
agent $r_i$ is allowed to cover as much of the segment $[s_i, s_{i+1}]$   
as it can, which determines the position $\ell_{i+1}$. Then, a recursive call to \textsc{Path} is made with arguments $i+1$, ~$\ell_{i+1}$ and the resulting energy deficit or surplus $tr_{i+1}$. 
For an easier understanding of the algorithm, the description below is annotated in detail and  Figure~\ref{figure:line} provides an example for each case encountered in the algorithm.

\begin{figure}[htb]
\centerline{
\includegraphics[width=4.8in]{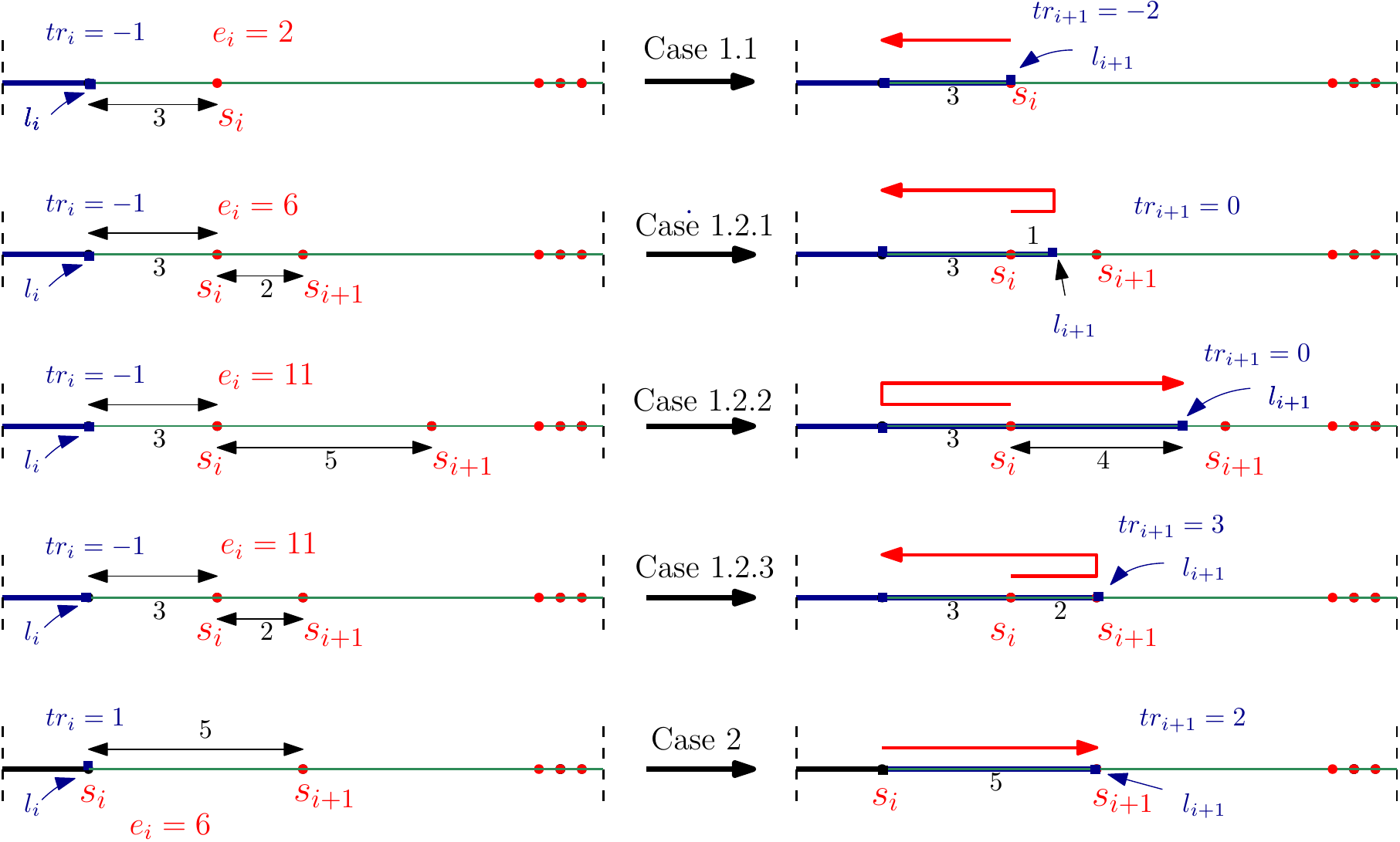}}
\caption{Assignment of a trajectory to  agent $r_i$ by Procedure \textsc{Path}.
The trajectory established in each case is in red on the right part of the figure. Cases 1.1 to 1.2.3 deal with an energy deficit prior to an assignment of a trajectory  to $r_i$, and  Case 2 deals with  a surplus  of energy. In Cases 1.2.1 to 1.2.3 the deficit is eliminated.}
\label{figure:line}
\end{figure}

We use the following lemma to simplify the proof of the main theorem of this section. 

\begin{lemma} \label{lm:linesegment}
Let $t_1, t_2, \ldots, t_{i-1}$ be the trajectories  and $\ell_i$, $tr_i$ be the values established by Procedure \textsc{Path} after $i-1$ recursive calls, $0\leq i \leq k+1$. These trajectories
explore the segment $[0,\ell_i]$ using in total energy $(\sum_{j=1}^{i-1}e_j)-tr_i$ and this is the minimum energy required by the agents to explore the segment $[0, \ell_i]$.
Furthermore, if $tr_i<0$ (deficit)  then $\ell_i=s_{i-1}$, if $tr_i>0$ (surplus) then   $\ell_i=s_i$, and if  $tr_i=0$ then $s_{i-1}\leq \ell_i\leq s_{i}$.
\end{lemma}
\begin{proof} The proof is done by induction on $i$ and omitted for the lack of space. \qed
\end{proof}
\begin{theorem} \label{th:linesegment1}
Assume we are given mobile agents $r_1,r_2,\ldots,r_k$ with energies  $e_1,e_2,\ldots,e_k$, located initially in positions
$s_1 \leq s_2 \leq \cdots \leq s_k$ of a line segment $[0,1]$ respectively.
Let $r_{k+1}$ be an additional  ``dummy agent'' at position 1 with zero energy.  
Then the procedure call \textsc{Path}$(1,0,0)$ on this instance runs in $O(k)$ time and terminates with \textbf{solvable} being true if and only if
there are trajectories of agents that explore the line segment $[0,1]$.
\end{theorem}
\begin{proof}
It is clear from the description  of the algorithm that it is linear in $k$. When the algorithm terminates with \textbf{solvable} being true, it is straightforward to see that a schedule can be determined for the agents, that creates a {\em valid} trajectory for each agent. To wit, in Round 1 all agents that receive no energy follow their trajectory and do the energy transfers as calculated. Notice that energy is received by agents when in their initial positions. \\
In Round $i$ all agents that receive energy in Round $i-1$ follow their trajectory and do the energy transfers as calculated. Therefore, if the algorithm terminates with \textbf{solvable} being true, the exploration of segment $[0,1]$ is possible, and our algorithm returns valid trajectories for the agents to achieve this coverage.
Thus we only need to show that the exploration of the segment $[0,1]$ is not possible when the algorithm terminates with \textbf{solvable} being false. In fact, \textbf{solvable} is set to false in the algorithm only if either $tr_{k+1} < 0$ or $\ell_{k+1} < 1$. By Lemma~\ref{lm:linesegment},
 if $tr_{k+1}<0$ then we can cover the segment up to $l_{k+1}=s_k\leq 1$ but we need more than the given
$(\sum_{j=1}^{k}e_j)$ in energy.
 If  $tr_{k+1}=0$ and $\ell_{k+1}<1$ then $[0, \ell_{k+1}]$ is the maximum segment that can be explored by the agents using $(\sum_{j=1}^{k}e_j)$ energy. In both cases, the exploration  of the entire segment $[0, 1]$ is impossible with the given initial positions and energies.
\qed
\end{proof}

Consider the case when we apply our algorithm to an input instance with the sum of energies $ \sum_1^ke_i\geq \frac{3}{2}$.
Since in each trajectory the part covered doubly is less or equal  the part covered singly, the  energy deficit/surplus $tr_{k+1}$ obtained after  assigning a trajectory to agent $r_k$  is at most $\sum_1^ke_i -\frac{3}{2} $, and it cannot be negative. Also $\ell_{k+1}$ cannot be less than 1 since there would be an unused surplus of energy of at least
$\sum_1^ke_i -\frac{3\ell_{k+1}}{2}>0$. Thus with $\sum_1^ke_i \geq \frac 3 2$ the algorithm terminates with valid exploration trajectories for the segment.

 On the other hand when the input instance consists of a single agent $r_1$ located at point $0.5$, the energy needed by $r_1$ to cover the segment $[0,1]$ is equal to $\frac{3}{2}$.
\begin{corollary} \label{th:linesegment2}
The segment $[0,1]$ can always be explored by $k$ agents with canonical trajectories if the sum of their initial energies is at least $\frac{3}{2}$, but exploration may be impossible in some instances if the sum is less than $\frac{3}{2}$.
\end{corollary}
\begin{remark} 
If agents {\em cannot share energy}, regardless of $k$, exploration is impossible in an input instance where all $k$ agents are co-located at $0$, each with energy equal to $1-\epsilon$, which gives total energy greater than $k-1$. Thus, without energy-sharing, there is no upper bound on the total energy of agents that guarantees exploration of a path. 
We have also constructed a linear algorithm for the exploration of a path by agents that cannot share energy, however we cannot include it in this paper due to the limit on the number of pages. 
\end{remark}

\section{Exploring a Tree}
\label{sec:Exploring a Tree}
%
In this section we consider a restricted case of the graph exploration problem, specifically when the graph is a tree. 
First we observe that,  {\em without energy exchange}, there is  a straightforward reduction from the partition problem showing that the exploration is NP-hard even on a star graphs: Given an instance of the partition problem $S = \{a_1, a_2, \ldots, a_n \}$, let $T = \sum_{i=1}^n a_i/2$. We  construct a star graph, with $n+2$ edges incident on the central node. Of these, $n$ edges have weight $a_1/2, a_2/2, \dots, a_n/2$ respectively, and two additional edges each have cost $T$. Assume two agents are at the central node of the star graph with energy $3T/2$ each. Then there is a partition of the set $S$ if and only if the there is an exploration strategy (without energy sharing) for the two agents on the star graph.  
However, for energy-sharing agents located on a tree we derive below a polynomial exploration algorithm (see also \cite{fraigniaud2006collective}).

Let $T$ be an edge-weighted tree with $k$ agents distributed across the $n$ nodes of $T$ with possibly several agents per node each of them with some non-negative energy. To simplify the design of our algorithm, we first preprocess the tree to transform it to a rooted binary tree where all the agents are located only at the leaves of the tree. We obtain such a tree from the initial tree $T$ in four stages: (a) by taking all the agents at every non-leaf vertex $v$ and shifting them to a new leaf node $l_v$ that is added to the tree and connected to $v$ via a zero-weight edge, (b) by repeatedly splitting vertices of degree more than $3$ into trees of maximum degree $3$ using zero-weight edges, (c) by collapsing any path with internal vertices of degree 2 into an edge whose weight equals the cumulative weight of the path, and (d) by converting the resulting $3$-regular unrooted tree into a rooted one by splitting one of its edges and making its midpoint be the root of the tree. Without loss of generality, we will denote this new, rooted tree as $T$. We note these preprocessing steps have complexity $O(n+k)$. Our problem can be stated as follows:

\begin{problem}[Tree Exploration]
Let $T$ be a rooted, edge-weighted binary tree obtained by preprocessing an unrooted edge-weighted tree with $k$ agents at its nodes, so that all the agents are now located at leaves in $T$ and have their given initial energies. For every node $v$, let $a_v$ be the initial number of agents inside subtree $T_v$ and let $e_v$ be the sum of their initial energies. Let $w_e\geq 0$ be the weight of edge $e$. If possible, find a set of valid trajectories for the agents that explore every edge of $T$ using only the given initial energies.
\end{problem}

Now, consider any \textit{feasible} exploration for the tree witnessed by a set of trajectories for the agents. The successful exploration of a subtree $T_v$ may either have necessitated additional energy brought into the subtree from outside, or there may be a surplus of energy that could have gone out of the subtree to explore other parts of the tree. Also, exploration may have needed a transfer of agents into the subtree (over and above its $a_v$ agents) or may have been accomplished with some agents made available to leave the subtree and contribute to the exploration of the rest of the tree.

We formalize this idea as follows. Let $B[v,i]$ denote the \textit{maximum} possible total \textit{surplus energy} that can \textit{leave} the subtree $T_v$ after it is fully explored so that $i$ agents can \textit{leave} the subtree with this total energy. Note that $i$ counts only the \textit{balance} of agents that depart from the tree, not individual arrivals and departures. Thus, we allow for $i$ being negative (i.e., $i$ agents enter the the tree on balance), or $B[v,i]$ being negative (i.e., $-B[v,i]$ amount of energy is needed to be brought in from \textit{outside} $T_v$ to explore it fully). We remark that:
\begin{enumerate}[label=(\roman*)]
    \item when $i\leq0$ and $B[v,i]\geq 0$, it means that an agent carrying excess energy \textit{must} leave $T_v$ but nevertheless, the overall balance of agents entering/leaving $T_v$ is non-positive.
    \item the $B[v,i]$ values do not take into consideration the energy expenditure that would be required to explore the edge from $v$ to its parent node in the tree.
    \item for node $v$, the value of $i$ can only be in the interval $[-k+a_v, ~a_v]$, because, on balance, at most $a_v$ agents can leave $T_v$ and at most $k-a_v$ can enter it.
\end{enumerate}


In order to simplify the calculation of $B[]$, we extend the definition of $B[]$ to edges as well: Let $e=(u,v)$ where $u$ is the parent of $v$. As described above, $B[v,*]$ denotes the surplus energy leaving subtree $T_v$.   $B[e,i]$ will denote the surplus energy available at $u$ \textit{along edge} $e$ with a balance of $i$ agents that could transit through $u$ from the direction of $v$. We observe that agents \textit{do} spend energy traversing $e$ itself and can also stop in the middle of $e$, and hence the values $B[e,*]$ and $B[v,*]$ are different. Since node $u$ has exactly two child edges below it, we can compute the $B[u,*]$ values by suitably combining the $B[]$ values of these edges.

It remains to show how to calculate $B[v,*]$ and $B[e,*]$ for each $v$ and $e$. In principle, we can consider all the possibilities of what the agents can do, but in reality it is enough to consider only the best possible activity with the desired balance of agents. The calculation is performed by calling procedures \HandleVertex~and \HandleEdge (described below in pseudocode) respectively for each vertex and for each edge of $T$. The computation proceeds in a bottom-up manner starting from leaves, with the boolean arrays \textbf{$\D[v]$} and \textbf{$\D[e]$} being used to ensure this flow.
In order to simplify the presentation in \HandleEdge, we assume that the assignment $B[e,i] \gets y$ is shorthand for $B[e,i] \gets \max (B[e,i], ~y)$ (with the initial values $B[e,*]$ being initialized to $-\infty$).

Our main result in this section is the following:
\begin{theorem}\label{thm:trees}
After transforming an unrooted tree with a specified distribution of initial agent locations and energies, procedure \HandleVertex, when applied to the root of the resulting rooted binary tree, correctly solves the tree exploration problem for the original tree in $O(n+ \ell k^2)$ time. 
It does so by correctly computing the optimal $B[v,i]$ values for every vertex $v$ and all relevant $i$ values for that vertex, in conjunction with procedure \HandleEdge~ that correctly computes the optimal $B[e,i]$ for every edge $e$ and all relevant $i$ values for that edge.
\end{theorem}

\vspace*{-3mm}
\begin{algorithm}
\caption{\textsc{HandleVertex}($v$)}
\label{alg-handleVertex}
{\normalsize
\begin{algorithmic}[1]
\If {$v$ is a leaf}
    \For {$i=-k+a_v$ to $a_v$}
       \State $B[v,i] \gets e_v$ \label{ln:leafs}
    \EndFor
\Else \Comment{$v$ has two child edges $e'$ and $e''$}
    \State wait until $\D[e']$ and $\D[e'']$
    \For {$i= -k+a_v$ to $a_v$} \label{ln:2ndfor}
        \State $B[v,i] \gets \max_{i'+i''=i} (B[e', i']+B[e'', i''])$ \label{ln:combine}
    \EndFor
\EndIf
\State $\D[v] \gets $\textbf{true}
\If {$v$ is the root}
    \If {there is $B[v,i]\geq 0$ with $i\geq 0$}
        \State return \textbf{solvable}
    \Else
        \State return \textbf{not solvable}
    \EndIf
\EndIf
\end{algorithmic}
}
\end{algorithm}

\begin{algorithm}[h]
\caption{\textsc{HandleEdge}($e$)}
\label{alg-handleEdge}
{\normalsize
\begin{algorithmic}[1]
\Require $e=(u,v)$ where $v$ is a child of $u$
\State wait until $\D[v]$
\For{$i=-k+a_v$ to $a_v$}
    \If {$B[v,i]\leq 0$}
        \If {$i<0$}
            \State $B[e,i] \gets B[v,i] - |i|w_e$  \label{ln:a} \Comment{Case 1a}
        \Else
            \State $B[e,i] \gets B[v,i] - (i+2)w_e$  \label{ln:d} \Comment{Cases 2a, 3a and 4a}
        \EndIf
    \ElsIf {$i\leq 0$} \Comment{$B[v,i] > 0$}
        \If {$B[v,i] > (2+|i|)w_e$}
            \State $B[e,i] \gets B[v,i]-(2+|i|)w_e$ \label{ln:c} \Comment{Cases 1c and 2c}
        \ElsIf {$i<0$}
            \State $B[e,i] \gets  i(w_e-B[v,i]/(2+|i|))$ \label{ln:b} \Comment{Case 1b}
        \Else \Comment{$i=0$}
            \State $B[e,-1] \gets (B[v,0]-2w_e)/2$ \label{ln:e} \Comment{Case 2b}
            \State $B[e, 0] \gets B[v,0]-2w_e$ \label{ln:e'}\Comment{Case 2b'}
        \EndIf
    \Else \Comment{$i>0$ and $B[v,i]>0$}
        \If {$B[v,i] \geq iw_e$}
            \State $B[e,i] \gets B[v,i]- iw_e$  \label{ln:i} \Comment{Cases 3c and 4c}
        \Else \Comment{$i>0$ and $B[v,i]<iw_e$}
            \State $B[e,i] \gets -(i+2)(w_e-B[v,i]/i)$  \label{ln:k} \Comment{Cases 3b'' and 4b}
            \If {$i=1$}
                \State $B[e,-1] \gets B[v,1]-w_e$ \label{ln:h} \Comment{Case 3b}
                \State $B[e, 0] \gets 2(B[v,1]-w_e)$ \label{ln:h'} \Comment{Case 3b'}
            \EndIf
        \EndIf
    \EndIf
\EndFor
\State $\D[e] \gets $\textbf{true}
\end{algorithmic}
}
\end{algorithm}
\vspace*{-3mm}

The proof of Theorem~\ref{thm:trees} hinges on an inductive argument that shows that the $B[]$ values are
correctly computed in a bottom-up manner starting at the leaves and working our way up the tree.
The base case for the induction is for the leaf nodes, and follows directly from the construction (line~\ref{ln:leafs} of \HandleVertex: all the energy in a leaf node is surplus and can be utilized to explore the rest of the tree).


%

Our induction hypothesis is established by proving two concomitant lemmas, the proofs of which are in the appendix:
\begin{lemma}\label{lm:Be}
Let $e=(u,v)$ be an edge in $T$ with $u$ being the parent of $v$. If the values $B[v,*]$ have been correctly computed, then procedure \HandleEdge ~correctly computes $B[e,*]$, where $*$ stands for all relevant values of $i$.
\end{lemma}

\begin{lemma}\label{lm:Bv}
Let $v$ be an internal vertex with two child edges $e'$ and $e''$. If $B[e',*]$ and $B[e'',*]$ have been correctly computed, then procedure \HandleVertex ~correctly computes $B[v,*]$ where $*$ stands for all relevant values of $i$.
\end{lemma}

It is easy to see after the $O(n+k)$ time preprocessing step to convert the original tree into a full binary tree, each leaf of the tree and subsequently, each edge of the tree can be processed in $O(k)$ time. To obtain the $B[]$ values for any internal node, we need to combine the $k$-vectors associated with the child edges (in step~\ref{ln:combine} of \HandleVertex). Observe that there are two classes of internal nodes in the converted tree. The first class correspond to $O(k)$ nodes that contained a subset of agents in the original tree. Each such node, generated in stage (a) of the conversion phase, has a child that is a leaf in the converted tree (containing the same subset of agents). The second class contains the internal nodes generated in stage (b) of the conversion phase, as well as the nodes that were of degree at least 3 in the original tree, and the root obtained in stage (d). Observe that there are at most $\ell$ nodes in the second class. Let $v$ be an internal node of the first class and let $v'$ be its child that was added in preprocessing that has taken the $v$'s agents. By construction, all the $B[v', *]$ values are equal to $e_{v'}$, and since the weight of $e'$ is $0$, $B[e', *] = e_{v'}$ as well. Hence, for such a node $v$ we have $B[v,a_{e'}+a_{e''}-i] = e_{v'}+ \max_{j=0}^i B[e'', a_{e''}-i]$ which can be for $i=0$ to $k$ computed in time $O(k)$. Consequently, the overall complexity of \textsc{HandleVertex} called for all $O(k)$ internal nodes of class one is $O(k^2)$. On the other hand, for each of $O(\ell)$ internal nodes from class two, the max function from step~\ref{ln:combine} of \textsc{HandleVertex} may be computed in $O(k)$ time, which leads to $O(\ell  k^2)$ overall complexity of all the calls of the \textsc{HandleVertex} procedure for nodes of the second class.   Applying Lemma~\ref{lm:Bv} to the root of the tree, it is clear that the algorithm correctly decides whether or not the binary tree can be explored fully, and the computation takes  $O(\ell k^2)$  time after the preprocessing step. 

We further remark that if exploration is indeed feasible then in the same time complexity, standard post-processing, top-down techniques can be used to recover the trajectories (and the schedules) for the agents from the computed $B[]$ values by combining the schedules locally computed at each vertex and edge. The choice of the root vertex is arbitrary (after preprocessing) and does not influence the decision outcome -- however, it does influence the computed schedule and the amount of energy left in the root, and there may be multiple feasible solutions (one for each $i \geq 0$ in the root's $B[]$ values). 
This completes the proof of Theorem~\ref{thm:trees}.

%
%
%


\section{General Graphs}
\label{sec:General Graphs}

Unfortunately, while the exploration problem for segments and trees admits very efficient solutions, for general graphs, exploration becomes intractable (unless P=NP). Indeed, we show that the graph exploration problem is NP-hard even in the case of 3-regular graphs by using a reduction from the Hamiltonian cycle problem. We also give an approximation algorithm.
\subsection{NP hardness for 3-regular graphs}
Let $G$ be a 3-regular graph on $n$ nodes.
We construct a graph $M$ by replacing each edge $e=(u,v)$ of $G$ by a meta-edge gadget $m(u,v)$ from Figure~\ref{fig:edgeGatget} case a), where $a$ and $b$ are chosen so that $a > 5nb$ where $n$ is the number of vertices of $G$. In addition, each meta-vertex (i.e., an image $m(v)$ of a vertex $v$ of $G$) starts with one agent and $3a+5b$ energy.

Note that the overall energy is $(3a+5b)n$, while the overall weight of the edges ($3n/2$ of them in a $3$-regular graphs) is $3n(a+b)$. Hence, a length at most $2bn$ can be crossed twice.  As $a>5bn$, this means no $a$-edge is crossed twice and at most $2n$ of $b$-edges are crossed twice.

\begin{figure}[htb]
\centerline{\includegraphics[width=5in]{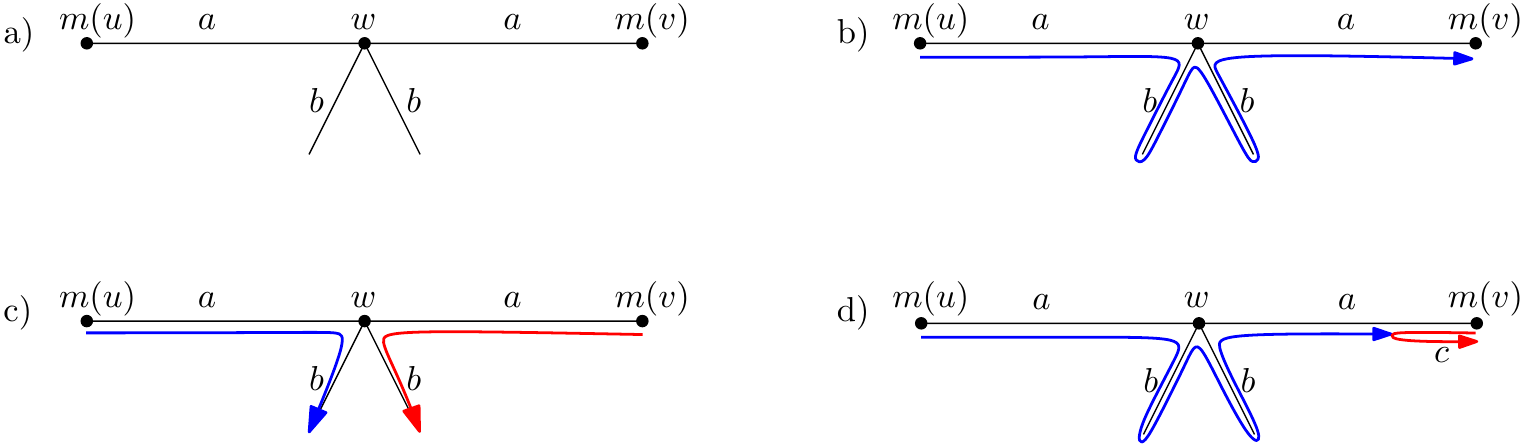}}
\caption{a) the meta-edge gadget, b) covering gadget using one agent c) efficient covering of the gadget using two agents d) covering gadget using two agents so that at least one agent exits the gadget}
\label{fig:edgeGatget}
\end{figure}

\begin{lemma}\label{lm:1agent}
If only one agent $x$ enters (w.l.o.g. from $m(u)$) a meta-edge $e = m(u,v)$, a total of $2a+4b$ energy is spent ensuring $e$ is fully explored. In such case $x$ ends up in $m(v)$.
\end{lemma}
\begin{proof}
Consult Figure~\ref{fig:edgeGatget}, case b). Let us call the edges of weigh $b$ starting at $w$ {\em whiskers}. As the agent has to cover both $(m(u),w)$ and $(w, m(v))$, since there is not sufficient energy to cover either of these edges twice, the agent must cover the whiskers before arriving at $m(v)$. Hence, both whiskers must be traversed twice, for the total energy expenditure of at least $2a+4b$, which the agent must have had at the moment it left $m(u)$.
\qed
\end{proof}

\begin{lemma}\label{lm:2agents}
Assume two agents $x$ and $y$ enter a meta-edge $e=m(u,v)$ from $m(u)$ and $m(v)$, respectively, and ensure $e$ is fully explored. If no agent leaves $e$ then the total energy spent in $e$ is at least $2a+2b$. If one or both agents leave $e$ then all the leaving agents leave to the same meta-vertex, and the total energy spent in $e$ is more than $2a+4b$.
\end{lemma}
\begin{proof}
If each of $x$ and $y$ have at least energy $a$, they can both arrive at the middle vertex $w$. If their total energy is at least $2a+2b$, they can share energy at $w$ so that each can complete one whisker (see Figure~\ref{fig:edgeGatget}, case c)). In order for at least one of the agents to leave $e$ (w.l.o.g. assume via $m(v)$), without having a full double-crossing of the edge $(w, m(v))$, the two agents must meet inside $(w, m(v))$, with one or both of them going to $m(v)$. No agent can return to $w$, otherwise the whole $(w, m(v))$ would have been double covered, for which there is no enough energy in the whole system. This means both whiskers must have been covered by the agent originating in $m(u)$ at a full double-cost of $4b$, for a total of more than $2a+4b$ (the part between the meeting point and $m(v)$ must be traversed twice, raising the cost sharply above $2a+4b$).
\qed
\end{proof}

Note that the second case of Lemma~\ref{lm:2agents} is worse than Lemma~\ref{lm:1agent} in terms of total energy expenditure, due to the extra cost of $c$ (or $2c$). Still, it might be justified if the blue agent does not have enough energy for case b), or if the red agent does not have enough energy for case c).

Lets call a meta-edge {\em light} if it is covered by the first part of Lemma~\ref{lm:2agents} (Figure~\ref{fig:edgeGatget}, case c)), otherwise it is {\em heavy}. Observe that each light meta-edge consumes $2$ agents, while each heavy meta-edge consumes an excess of $2b$ energy compared to the weight of its edges. This yields:

\begin{lemma}\label{lm:heavy}
If there is an exploration strategy for the input graph, the number of heavy edges is exactly $n$.
\end{lemma}
\begin{proof}
As the total energy in the system exceeds the total weight of the edges by $2bn$ and each heavy meta-edge consumes $2b$ more energy than the weight of its edges, there can be at most $n$ heavy meta-edges. On the other hand, as there are only $n$ agents available in total and each light meta-edge consumes two agents, there can be at most $n/2$ light meta-edges. As the total number of meta-edges is $3n/2$, the number of heavy edges must be exactly $n$.
\qed
\end{proof}

By Lemma~\ref{lm:1agent} and the second part of Lemma~\ref{lm:2agents}, a heavy meta-edge is traversed in one direction -- lets call the halves {\em outgoing} and {\em incoming}.

Consider the directed graph $H=(V, E')$ formed by the heavy meta-edges, i.e., $e=(u,v)\in E'$ iff $e$ is a heavy meta-edge from $m(u)$ to $m(v)$.

\begin{lemma}\label{lm:incoming}
Each vertex of $H$ has at least one incoming edge.
\end{lemma}
\begin{proof}
As both light and outgoing edges consume agents, if $v$ had no incoming edge, it would need $3$ agents.  Since each vertex starts with $1$ agent, it is not possible for $v$ to have $3$ agents without an agent coming via an incoming edge.
\qed
\end{proof}

Because the number of heavy edges is exactly $n$, the heavy edges form a vertex-disjoint (vertex) cycle cover of $H$, i.e., each meta-vertex has one incoming, one outgoing and one light edge.

The problem is that the disjoint cycle cover is solvable in polynomial time. Hence, we need to modify the input so that there is a solution to the exploration problem if and only if the heavy edges form a single cycle of length $n$.

This is done by modifying the input into $I$ as follows:
\begin{itemize}
\item the three edges incident to the initial vertex have weights adjusted to $a+\epsilon n$ for some small $\epsilon < b/3n$,
\item the energy at the initial vertex is $3a+5b+(2n)\epsilon$
\item the energy at all other meta-vertices is $3a+5b+\epsilon$
\end{itemize}

\begin{lemma}\label{lm:if}
If the graph $G$ is Hamiltonian then there is an exploration solution to the modified input $I$.
\end{lemma}
\begin{proof}
Select the direction of the Hamiltonian cycle, its edges will be the heavy meta-edges.
The {\em explorer} agent starting in the initial vertex takes all the  energy available there and follows the Hamiltonian cycle. It collects $2a+4b+\epsilon$ energy in each of the meta-vertices it crosses, while spending $2a+4b$ on each heavy meta-edge. The agents located at other meta-vertices wake-up when the explorer arrives, take $a+b$ energy and explore half of the incident light meta-edge.

Note that when the explorer reaches $i$-th meta-vertex (not counting the initiator), it has $a+b+\epsilon n+i-1)$ energy remaining, except when it returns to the initiator, when it has only $a+b+\epsilon n$ energy as the last meta-edge it crossed has an extra $\epsilon n$ cost. This is just sufficient to cover half of the incident light meta-edge, which is the last part not yet covered.
\qed
\end{proof}

\begin{theorem}\label{th:onlyIf}
The exploration problem is NP-hard for 3-regular graphs. 
\end{theorem}
\begin{proof}
Suppose there is an exploration solution for $I$. We claim that then the graph $G$ is Hamiltonian. Note that since the sum of all  $\epsilon$ is less than $b$, Lemma~\ref{lm:heavy} and Lemma~\ref{lm:incoming} still hold. Hence, the only way meta-edges incident to the starting vertex can be covered is if the agent returning to the starting vertex carries $a+b+\epsilon n$ energy. As the agent can gain only $\epsilon$ energy in each vertex it crosses (the remainder is used-up on crossing the heavy meta-edge and covering half of the incident light meta-edge), it needs to visit all $n$ meta-vertices in order to collect sufficient energy, i.e., it has performed Hamiltonian cycle. The theorem then follows since the Hamiltonian cycle problem for 3-regular graphs is known to be NP-complete.
\qed
\end{proof}

\subsection{An Approximation Algorithm for General Graphs}

Even though the graph exploration problem is NP-hard, it is still possible to obtain an efficient approximation algorithm for exploring arbitrary graphs that has an energy-competitive ratio at most $2$. Specifically, the algorithm uses at most twice as much energy as the cumulative sum of the edge weights of the graph.
First we state without proof a well-known result for agents on a cycle for which the reader is referred to \cite{lovasz2014combinatorial}[Paragraph 3, Problem 21].


\begin{lemma}
\label{lm:cycle}
For any cycle and any initial positions of the agents there is an algorithm which explores the cycle if and only if  the given  sum of the energies of the agents is not less than the length of the cycle.
\end{lemma}

Lemma~\ref{lm:cycle} has some important consequences. Recall that a graph $G$ is Eulerian if it is connected and all its vertices have even degrees.

\begin{theorem}
\label{thm:euler}
For any Eulerian graph and any initial positions of the agents there is an algorithm which explores the graph provided that the sum of the energies of the agents at the start is not less than the sum of the edges of the given graph. Moreover, the algorithm has optimal energy consumption.
\end{theorem}
\begin{proof}
Assume we have agents in such a graph so that the sum of the energies of the agents is at least equal to the sum of the length of the edges of the graph. Since the graph is Eulerian we can construct a cycle which traverses all the edges of the graph exactly once.  By Lemma~\ref{lm:cycle} there is an algorithm which assigns trajectories to the given sequence of agents and covers the entire graph. This proves Theorem~\ref{thm:euler}.
\qed
\end{proof}

\begin{theorem}
\label{thm:graph}
Any graph can be explored  by energy-sharing agents if the sum of their initial energies is at least twice the 
sum of edge weights. Moreover, this constant $2$ cannot be improved even for trees. 
\end{theorem}
\begin{proof}
The original graph, say $G = (V,E)$, is not necessarily Eulerian. However, by doubling the edges of the graph we generate an Eulerian graph $G' = (V, E')$. The sum of the weights of the edges of $G'$ is equal to twice the sum of the weights of the edges of $G$. Theorem~\ref{thm:euler} now proves that the graph $G$ can be explored if the sum of their initial energies is at least twice the sum of edge weights in the graph.

Consider a star graph with $2k$ leaves, all edges are of weight 1, for a total weight of $2k$. Assume that we have $k$ agents in the leaves of the star, and one agent in the center of the star. 
Agents in the leaves have energy $0$, and the agent in the center
$4k-1$. 
To explore the graph  the center agent can traverse $2k-1$ edges of the star  twice and one edge once.  It is easy to see that no other strategy can do it with less energy.
Thus the energy 
of the middle agent cannot be any lower, and asymptotically, 
total energy in this instance equal to $4k-1$ approaches  the double of the cost of the edges. 
  \qed
\end{proof}

\begin{remark}
An  improvement of the competitive ratio $2$ of Theorem~\ref{thm:graph} is possible in specific cases by using a Chinese Postman Tour~\cite{chinese} (also known as Route Inspection Problem). Namely, in polynomial time we can compute the minimum sum of edges that have to be duplicated so as to make the graph Eulerian,  and this additional sum of energies is sufficient for the exploration.
\end{remark}


Assume we have a fixed configuration $C$ of $k$ agents
$r_1,r_2,\ldots,r_k$ in a given graph $G$.
We say that energy assignment $E=e_1,e_2,\ldots,e_k$ to agents in
configuration C is {\em minimal} if exploration is possible with
energies in $E$, but impossible when the energy level of any one agent
is decreased. Let $|E|=\sum_{i=1}^k e_i$.
We now investigate how large  the ratio $|E_1|/|E_2|$  can be 
for two minimal assignments $E_1$ and $E_2$ of a given configuration $C$,


  

\begin{claim}
For any configuration of agent $C$  and any two minimal energy assignments  
$E_1$ and $E_2$ the ratio  $|E_1|/|E_2|\leq 2$, and this is asymptotically optimal.  
\end{claim}
\begin{proof} See the appendix.
\qed
\end{proof}

\section{Conclusion}
\label{sec:conclusion}

 We studied graph exploration by a group of mobile agents which can share energy resources when they are co-located. We focused on the problem of deciding whether or not it is possible to find trajectories for a group of agents initially placed in arbitrary positions with initial energies so as to explore the given weighted graph. The problem was shown to be NP-hard for $3$-regular graphs  while for general graphs it is possible to obtain an efficient approximation algorithm that has an energy-competitive ratio at most $2$ (and this is shown to be asymptotically optimal). We also gave efficient algorithms for the decision problem for paths, trees, and Eulerian graphs.
 The problem considered is versatile and our study holds promising directions for additional research and interesting open problems remain by considering exploration with 1) optimal total energy consumption, 2) agents with limited battery capacity, 3) energy optimal placement of mobile agents, 4) time vs energy consumption tradeoffs for mobile agents with given speeds, 5) additional topologies, as well as 6) combinations of these.


\bibliographystyle{plain}
\bibliography{refs,refs1}
\newpage
\section{Appendix}

\begin{proof}{\bf of Lemma \ref{lemma:l1}}\\
First we show that the trajectories can be modified to satisfy part (i).
Let $t_i$ be a trajectory of agent $r_i$ which receives some energy from its neighbour at an endpoint different from $s_i$.
Assume first that in trajectory $t_i$ agent $r_i$ moves to the right of $s_i$ and it arrives at $b_i^r$ with $x_1$ being its remaining energy,  it receives there additional energy $x_2$ from $r_{i+1}$, and the remaining energy of $r_{i+1}$ at $b_i^r$ after the transfer is $x_3$.
Let $y=b_i^r-s_i$. There are two cases to consider.

Case 1: $x_2 \geq 2y$. Clearly, the energy available to $r_i$ at $b_i^r$ is  sufficient to move back to $s_i$.  Consider the situation after  $r_i$ moves
back to $s_i$  from $b_i^r$. At that point
the energy of $r_i$ is $x_1+x_2-y$. Consider modified trajectories $t_i'$ and $t_{i+1}'$ for $r_i$ and $r_{i+1}$ as follows:
Agent $r_i$  remains at $s_i$, and agent $r_{i+1}$ moves to $s_i$, and the energy transfer is done at $s_i$ by transferring to $r_i$ the energy $x_2 - 2y$.
Notice that $r_{i+1}$ has sufficient energy to reach $r_i$ and go back to $b_i^r$ since $x_2\geq 2y$, in case the trajectory $t_{i+1}$ of $r_{i+1}$ was
turning back at left.  Then the energy of $r_i$ at point $s_i$ after the energy transfer would be $x_1+y +x_2 - 2y=x_1+x_2-y$, and $r_{i+1}$ after returning to $b_i^r$ is left with energy $x_2+x_3-2y -(x_2 - 2y)=x_3$, i.e., both agents are exactly with the same energy as in $t_i$, $t_{i+1}$ at the same points and they can follow the rest of their previous trajectories.

 Case 2: $x_2 < 2y$. Let $c=b_i^r-(x_2/2)$.
Consider the situation when  $r_i$ reaches $c$ from $b_i^r$. At that point
the energy of $r_i$ is $x_1+x_2/2$. Consider modified trajectories $t_i'$ and $t_{i+1}'$ for $r_i$ and $r_{i+1}$ as follows:
Agent $r_i$ moves from $x_i$ to $c$, where it turns to the left. Agent $r_{i+1}$ moves to $c$. If the trajectory of $r_{i+1}$ ended at $b_i^r$ then now it ends at $c$, else it turns right to go back to $b_i^r$.
There is no energy transfer  done at $c$.  In the modified trajectory the energy of $r_1$ at the turning point $c$ is $x_1+x_2/2$, and $r_{i+1}$ even if it needs to return to $b_i^r$ is left with energy $x_2+x_3 -2(x_2/2)=x_3$, i.e., both agents are exactly with the same energy level or better as in $t_i$, $t_{i+1}$ at the same points, and they can follow the rest of their previous trajectories and possible energy transfers as before.\\
Clearly a symmetric modification can be done if $r_i$ first moves to the left
of $s_i$.
Therefore,  by repeated application of the above transformation, we can ensure that the trajectories covering the segment $[0,1]$  satisfy $(i)$ of the lemma while using at most the given energies.

We then process every trajectory not satisfying $(ii)$ of the lemma as follows.
Let $t_i$ be a trajectory that does not satisfy $(ii)$ of the theorem.
Assume that trajectory $t_i$ of agent $r_i$ first goes from point $s_i$ to point $b_i^r>s_i$ on the right of $s_i$, then to point $b_{i}^{\ell}<s_i$ on the left of $s_i$, and
$|s_ib_i^r| > |s_ib_{i}^{\ell}|$. Since $t_i$ satisfies $(i)$, agent $r_i$ receives no energy  at $b_i^r$.
 Then, clearly, replacing $t_i$ by $t_i'$ which first goes left to
$b_{i}^{\ell}$ and then to $b_i^r$ we have a trajectory that is shorter, since the doubly covered part is shorter than the singly covered part. If $r_i$  delivers
some energy to $r_{i-1}$ or $r_{i+1}$, it can deliver the same amounts when following $t_i'$. This change does not influence any other trajectory.
Clearly a symmetric modification can be done if $r_i$ first goes to the left
of $s_i$.
\qed
\end{proof}

\begin{proof}{\bf of Lemma~\ref{lm:linesegment}}\\
We give a proof by induction on $i$.
Trivially, the inductive hypothesis is true  for $i=0$, before any trajectory is assigned, since $\ell_1=0,tr_1=0$.

Assume now that the statement is true for some $i$, where  $0 \leq i\leq k-1 $. We will show that the statement is true for $i+1$. We consider the following exhaustive cases:\\
{\bf Case 1}: $tr_{i}>0$, there is a surplus of energy in the initial part.  Then the trajectories assigned so far explore the interval up to $s_i$.
Clearly, the unused energy $tr_i$ should be transferred to $r_i$.
By Lemma \ref{lemma:l1} the transfer of energy is done at $s_i$,
agent $r_i$ will have energy $e_i+tr_i$, and we may assign to $r_i$ a straight line trajectory $t_i$ from $s_i$ to the right until either $r_i$ runs out of energy, or it reaches $s_{i+1}$, i.e., to point $s_{i}+y=l_{l+1}$ where
$y=\min\{s_{i+1}-s_i,s_i+e_i+tr_i\}$, and this is the trajectory assigned by our algorithm.   Clearly, a straight line trajectory of $r_i$ is energy optimal, and following this trajectory  $r_i$ ends with
unused energy $tr_{i+1}=e_i+r_i-y\geq 0$. Thus trajectories
$t_1, t_2, \ldots,t_{i-1},t_i$ explore the segment $[0,l_{i+1}]$ using in total energy $(\sum_{j=1}^{i}e_j)-tr_{i+1}$ and this energy cannot be made smaller, Furthermore, if $tr_{i+1}>0$  then   $l_{i+1}=s_{i+1}$, and if  $tr_{i+1}=0$ then $s_{i}\leq l_{i+1}\leq s_{i+1}$, and the inductive hypothesis is  true for $i+1$. \\
{\bf Case 2}: $tr_i<0$, there is a deficit in the initial part. Then the trajectories assigned so far cover the interval up to $\ell_i=s_{i-1}$. By Lemma \ref{lemma:l1} the transfer of energy is done by $r_i$ at $s_{i-1}$.
The trajectory of $r_i$ depends on its energy level:\\
If $e_i < tr_i+(s_i-s_{i-1})$  then $r_i$ does not have enough  energy to cover the sub-segment from $s_i$ to
$s_{i-1}$ and to give $r_{i-1}$ the needed energy and thus it needs to receive energy from $r_{i+1}$. By Lemma \ref{lemma:l1} the trajectory $t_i$ of $r_{i}$ should only go from $s_{i}$ to the left until $s_{i-1}$, and deliver the required energy to $r_{i-1}$, with the  deficit at $s_i$ being  $tr_i-(s_i-s_{i-1})$.  A straight line trajectory of $r_i$ is energy optimal and it is the trajectory assigned by our algorithm. Thus, with addition of $t_i$ trajectories
$t_1, t_2, \ldots,t_{i-1},t_i$ cover the segment $[0,s_{i}]$ using in total energy $(\sum_{j=1}^{i}e_j)-tr_{i+1}$ and this energy cannot be made smaller. Furthermore, $l_{i+1}=s_{i}$, and the inductive hypothesis is  true for $i+1$. \\
If $e_i \geq tr_i+(s_i-s_{i-1})$  then $r_i$  has enough  energy to cover the sub-segment from $s_i$ to
$s_{i-1}$, to give $r_{i-1}$ energy $tr_i$, and it has $e_i - (tr_i+(s_i-s_{i-1}))$ energy left over. This energy must be used to cover a part of the segment to the right of $s_i$, but not past $s_{i+1}$. Since $r_i$ moves left and right of $s_i$, a part of its trajectory is covered twice. By Lemma \ref{lemma:l1},  the part of the segment covered doubly should not exceeded the length the segment covered  singly. Let $e_i'=e_i-|tr_i|$, i.e., it is the energy available for covering a segment around $s_i$.\\
Assume  $s_i-\ell_i \leq e_1' \leq 3 (s_i-\ell_i)$ and
let $y = \min\{ \frac{e_i'-(s_i-\ell_i)}{2}, s_{i+1}-s_i \}$.  Consider a trajectory of  $r_i$ that goes right to  $\ell_{i+1}=s_i+y$ from $s_i$, then turns left and goes to $s_{i-1}$ where it gives energy $tr_i$ to $r_{i-1}$. In the traversal it spends energy $2y +(s_i-\ell_i)$. If $s_i+y=s_{i+1}$, then agent $r_i$ gives its excess energy $e_i-|tr_i|-(s_i-\ell_i) -2y$  to $r_{i+1}$  when at $s_{i+1}$.
Since $y\leq (e'-(s_i-\ell_i))/2\leq  (3(s_i-\ell_i)-(s_i-\ell_i)/2=(s_i-\ell_i)$, the length of the doubly covered part, is at most equal to  $s_i-\ell_i$, and thus the doubly covered part should be on the right.  This
 trajectory of $r_i$ covers the segment up to
$ s_i+y$ and with the available energy, it cannot cover more. This is the trajectory and energy transfers assigned by our algorithm.\\
If $ e_1' >3 (s_i-\ell_i) \mbox{ and } (s_i-\ell_i) \geq (s_{i+1}-s_{i}) $ then, clearly,  it is optimal to do double coverage on the right. Since the available energy is sufficient for it, trajectory of $r_i$ should go from $s_i$ all the way to $s_{i+1}$, there $r_i$ turns left and goes to $\ell_i$ to deliver energy $tr_i$ to $r_{i-1}$. When at $s_{i+1}$ it should give its surplus energy
$tr_{i+1}=e_i'-2(s_{i+1}-s_i)-(s_i-\ell_i)\leq e_i'$ to $r_{i+1}$.
 This  is the trajectory and energy transfers assigned by our algorithm. \\
If $ e_1' >3 (s_i-\ell_i) \mbox{ and } (s_i-\ell_i) < (s_{i+1}-s_{i}) $ then
$r_i$ has sufficient energy to cover the sub-segment from $\ell_i$ to $s_{i+1}$. Since $(s_i-\ell_i) < (s_{i+1}-s_{i})$, it is optimal to do a double coverage on the left.  The trajectory of $r_i$ thus should go left from $s_i$ to
$\ell_i=s_{i-1}$ where it gives energy $tr_i$ to $r_{i-1}$,
then turns right and goes to $s_{i+1}$ to deliver its remaining energy $tr_{i+1}e_i' -2(s_i-\ell_i)-(s_{i+1}-s_i)$ to $s_{i+1}$. This  is the trajectory and energy transfers assigned by our algorithm. \\
Thus in all these cases  trajectories
$t_1, t_2, \ldots,t_{i-1},t_i$ cover the segment $[0,l_{i+1}]$ using in total energy $(\sum_{j=1}^{i}e_j)-tr_{i+1}$ and this energy cannot be made smaller, Furthermore, if $tr_{i+1}>0$  then   $l_{i+1}=s_{i+1}$, and if  $tr_{i+1}=0$ then $s_{i}\leq l_{i+1}\leq s_{i+1}$, and the inductive hypothesis is  true for $i+1$. \\
{\bf Case 3}: $tr_i=0$, there is no deficit. This is similar to Case 2 and its sub-cases and thus omitted. \\
This completes the proof by induction.
\qed
\end{proof}
\begin{proof}{\bf of Lemma \ref{lm:Be}}\\
Before elaborating on the proof via a case analysis of possibilities, we first examine which scenarios for agents crossing the edge $e$ are energy efficient and need to be analyzed in detail, and which ones are wasteful and hence can be ignored.
\begin{figure}[htb]
\centerline{\includegraphics[width=8cm]{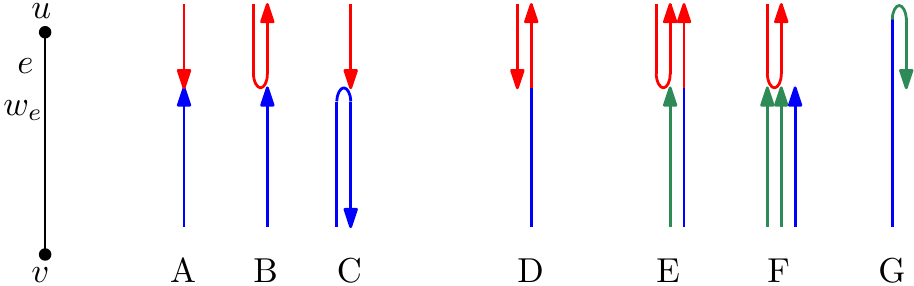}}
\caption{Optimizing scenarios with agents terminating inside edge $e$. The green arrows are the wasteful extra agents.}
\label{fig:edgeOptimization}
\end{figure}
As mentioned earlier, agents can collaborate to explore the edge, and this could mean that it is sometimes useful for an agent to expend all its remaining energy to explore the edge \textit{partially} starting from one endpoint of the edge, while leaving other agents to explore the rest of the edge. In particular, we are interested in justifying the cases when an agent really needs to terminate inside the edge.

Consider the cases in Figure~\ref{fig:edgeOptimization} where each blue/green arrow represents an agent exploring a part of the edge $e=(u,v)$ from $v$, while each red arrow represents an agent coming from $u$ to explore the rest of the edge. We claim that cases A, B and C from Figure~\ref{fig:edgeOptimization} are the only scenarios that make sense from the point of energy efficiency. Indeed, case D can be seen as the same as case B, with the agents switching their roles in their meeting point.  Hence, if an agent crosses the whole edge, having an agent traveling in the same direction that terminates inside $e$  (the green arrow in case E) serves no purpose at all: it only increases the net energy expenditure and hence the case can be ignored. Similarly, it is wasteful to have two or more agents going in the same direction and terminating inside of $e$, as the same effect (namely, covering that part of $e$, bringing energy to agents coming from the opposite side) can be achieved by a single agent as in case F. Finally, once an edge $e$ has been fully crossed (and hence explored), there is no point in turning back and terminating inside $e$ as that serves no purpose (case G). As this applies for both directions in $e$, the result is that in each direction at most one agent terminates inside $e$ and if it does so, there are no agents crossing $e$ in this direction. Hence, \textit{cases A, B and C are indeed all the scenarios with agents terminating inside $e$ that need to be considered}.



We are now in a position to describe all the possibilities that are relevant to calculating $B[e,*]$ for any edge $e=(u,v)$, namely those in Figure~\ref{fig:handleEdge}. Blue arrows represent surplus energy available from subtree $T_v$ (i.e., with $B[v,*] \geq 0$) while red arrows represent an energy deficit that has to be provided from $u$. The possibilities shown in the figure belong to four groups characterized by the value of $i$ in $B[v,i]$: group 1 where $i<0$ (the example $i=-2$ is shown), group 2 where $i=0$, group 3 where $i=1$ and group 4 for $i>1$ (the example $i=2$ is shown). Within each group, there are three basic cases based on $B[v,i] \leq 0$ (case a); $B[v,i] >0$ and with enough surplus energy to allow the $i$ agents to traverse the edge (case c); and $B[v,i] > 0$ but without enough energy to avoid having to terminate agents after partially exploring the edge (sub-cases b, b' etc.). These latter sub-cases are more nuanced for $i=0$ and $i=1$, which is the reason for having separate groups for them.
\begin{figure}[htb]
\centerline{\includegraphics[width=5in]{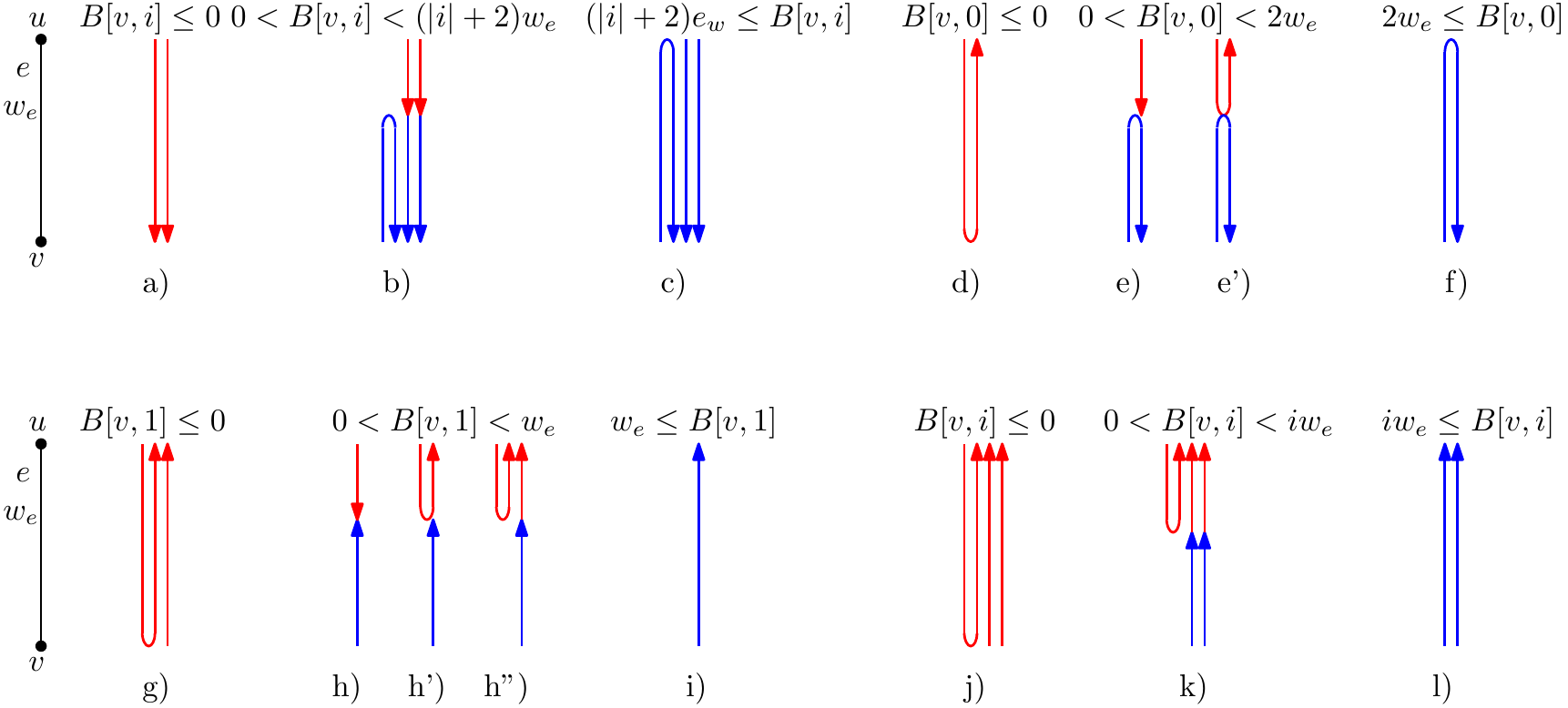}}
\caption{Handling an edge $e=(u,v)$, based on $B[v,i]$. The four groups are characterized by the value of $i$ in $B[v,i]$: group 1 where $i<0$ (the example $i=-2$ is shown), group 2 where $i=0$, group 3 where $i=1$ and group 4 for $i>1$ (the example $i=2$ is shown).}
\label{fig:handleEdge}
\end{figure}


Indeed, the assignment of $B[e,*]$ values in \HandleEdge ~closely tracks the sub-cases in Figure~\ref{fig:handleEdge}. Consider first the scenarios where no agent terminates inside the edge $e$. This means $B[e,i] = B[v,i]-x$ where $x$ is the cost incurred on the edge $e$. Let us do the case analysis to derive $x$:

\begin{itemize}
\item if $i<0$ and $B[v,i] \leq 0$, both $i$ agents and $B[v,i]$ energy must reach $v$ from $u$ at a cost of $x=|i|w_e$ for those $|i|$ agents to cross $e$ (see case 1a from Figure~\ref{fig:handleEdge} and line~\ref{ln:a} of \HandleEdge)

\item if $i<0$ and $0<B[v,i]$: $i$ agents must cross from $u$ to $v$. While the minimal cost of this crossing is achieved by $|i|$ agents crossing from $u$ to $v$ at a cost of $|i|w_e$, they would need to take energy from $u$, effectively wasting all of $B[v,i]$. Hence, the best strategy is to make the optimal use of the energy available in $v$, even if that means spending more in exploring $e$:
\begin{itemize}
\item If $B[v,i]>= (2+|i|)w_e$, there is enough energy in $v$ to allow an agent to bring this excess to $u$, pay $|i|w_e$ for the transit of $|i|$ agents from $u$ to $v$, and $2w_e$ for its own costs (as it also returns to $v$ to maintain the balance of $|i|$ agents). Hence, $x=(2+|i|)w_e$ (see cases 1c and 2c from Figure~\ref{fig:handleEdge} and line~\ref{ln:c} of \HandleEdge).

\item If $0<B[v,i]<(2+|i|)w_e$, there is not enough energy in $v$ to do all of this, and extra energy from $u$ is needed. As the total energy spent on $e$ is $|i|w_e+2y$ where $y$ is the distance traveled by the agent from $v$, it is minimized when $y$ is minimized. In turn, $y$ is minimized when $B[v,i]=(2+|i|y)$ so that when it meets the agents from $u$ it has just enough energy to bring them (and itself) to $v$ (otherwise some of the energy in $v$ would have been wasted). In such case, the extra energy brought from $u$ is equal to $|i|(w_e-y)$ (see case 1b from Figure~\ref{fig:handleEdge} and line~\ref{ln:b} of \HandleEdge).
\end{itemize}

\item if $i\geq0$ and $B[v,i] \leq 0$, an agent has to come from $u$ to $v$ to bring the energy so that $i$ agents can cross from $v$ to $u$, with the overall cost of $x=(2+i)w_e$ (see cases 2a, 3a and 4a from Figure~\ref{fig:handleEdge} and line~\ref{ln:d} of \HandleEdge).

\item if $i>0$ and $0 \leq B[v,i]$: $i$ agents must cross from $v$ to its parent.
\begin{itemize}
\item If $\geq iw_e$, they have enough energy to do so without help from the parent, incurring cost $x=iw_e$ (see cases 3c and 4c from Figure~\ref{fig:handleEdge} and line~\ref{ln:i} of \HandleEdge).

\item If $0 < B[v,i] < iw_e$, extra energy from $u$ is needed to help bring the $i$ agents all the way to $u$. Let $y$ is the distance traveled from $u$ by the agent bringing this extra energy. The optimal cost of $iw_e+2y$ is achieved by minimizing $y$, however $y\geq (w_e=B[v,i]/i)$ in order to reach all the agents that need energy. (see cases 3b'' and 4b from Figure~\ref{fig:handleEdge} and line~\ref{ln:k} of \HandleEdge).
\end{itemize}

\item if $i=0$, as the $e$ must be explored, the cost $x$ on it cannot be $iw_e=0$ but it must be $2w_e$ (still assuming no agent terminates inside $e$). The cases $B[v,i]<0$ and $B[v,i]>2w_e$ have already been dealt with, which leaves case 2b' from Figure~\ref{fig:handleEdge} and line~\ref{ln:e'} of \HandleEdge.
\end{itemize}

We are left with analyzing the cases when an agent (or two) terminates inside $e$. As shown before, the only such scenarios that are possibly optimal are cases A, B and C from Figure~\ref{fig:edgeOptimization}. These correspond to cases 3b, 3b' and 2b from Figure~\ref{fig:handleEdge} and lines~\ref{ln:h}, \ref{ln:h'} and \ref{ln:e} of \HandleEdge, respectively. It is straightforward to verify that the values computed in those lines indeed correspond to the energy supplied from $u$. Remember that $\gets$ assigns only if the new value is better (higher) than the previous one. Hence, $B[e,-1]$ at the end gets the best value from case 1b of $B[v,-1]$, case 2b of $B[v,0]$ and case 3b of $B[v,1]$ (provided the values of $B[v,*]$ with those cases exist) while $B[e,0]$ gets the better of case 2b' of $B[v,0]$  and case 3b' of $B[v,1]$.
\qed
\end{proof}

\begin{proof}{\bf of Lemma~\ref{lm:Bv}}\\
As no energy and agents are lost at $v$, the optimal energy for the balance of $i$ agents at $v$ is obtained by considering all possibilities how to combine agents from the subtrees to obtain $i$ agents and taking the best option (line~\ref{ln:combine} of \HandleVertex).

It remains to be shown that the schedule with balance $B[v,i]$ is indeed feasible (it might not be because a lack of agents or energy at the right time and place). Let $i'$ and $i''$ be the agent balances that result in the best $B[v,i]$. W.l.o.g. assume $i'\geq i''$. If $i'>0$ and $B[v,i']>0$, this means that $T_{e'}$ can be explored without any contribution from $v$ or its parent (by induction hypothesis), bringing the excess agents and energy to $v$, which has now exactly the right amount of agents and energy (with possible contribution from its parent) to complete the exploration of $T_{e''}$. This works even if $i'$ or $B[e', i']$ are not positive, as long as the contribution from the parent allows to explore $T_{e'}$ or $T_{e''}$.

However, there is one case where this does not work: If $i$ and $B[v,i]$ are positive, $i'$ and $B[e'', i'']$ are positive but $B[e', i']$ and $i''$ are negative. In such case there is no help coming from the parent of $v$, and neither subtree can be explored first on its own in a naive manner. However, the fact that $B[e'', i'']$ is positive means that there is an agent in $T_{e''}$ with excess energy. That agent first comes out to $v$ and brings energy to explore the $T_{e'}$ and then returns (together with the agents it brings from $T_{e'}$) to explore $T_{e''}$. This does not cost extra energy beyond what is already accounted for in $B[e'', i'']$ -- this is captured in case 1c from Figure~\ref{fig:handleEdge} and line~\ref{ln:c} of of \HandleEdge.
\qed
\end{proof}
\begin{claim}
For any configuration of agent $C$  and any two minimal energy assignments  
$E_1$ and $E_2$ the ratio  $|E_1|/|E_2|\leq 2$, and this is asymptotically optimal.  
\end{claim}
\begin{proof}
We know  from Theorem \ref{thm:euler} that $|E_1|/|E_2|\leq 2$. 
Consider a star graph with $2k$ leaves, all edges of weight 1.  Assume that we have configuration $C$ consisting of $k$ agents in the
leaves of the star, and one agent in the center of the star, 
and the following two energy assignments for $C$:\\
$E_1$:  agents in the leaves have energy $0$, and the agent in the center
$4k-1$ for a total $4k-1$, i.e.,  as in the proof of Theorem \ref{thm:graph}.\\
$E_2$: agents in the leaves have energy $2$, and the agent in the center
$0$ for a total $2k$.\\
$E_1$ was shown to be minimal in the proof of Theorem \ref{thm:graph}, and
$E_2$ is obviously minimal. 
Clearly, $|E_1|/|E_2|=(4k-1)/2k$  which converges to $2$
for large $k$. 
\qed
\end{proof}

\end{document}